\documentclass[journal]{IEEEtran}

\IEEEoverridecommandlockouts                              

\usepackage{empheq,amssymb}
\usepackage{braket}
\usepackage{dsfont,mathrsfs,bm}
\usepackage{enumerate}
\usepackage{float}
\usepackage{lipsum}
\usepackage{cite}
\usepackage{amsmath,amsthm}
\usepackage{amsthm}
\newtheorem{theorem}{Theorem}

\newtheorem{definition}{Definition}
\newtheorem{claim}{Claim}
\newtheorem{lemma}{Lemma}

\newtheorem{remark}{Remark}

\usepackage{mathtools}
\DeclarePairedDelimiter{\ceil}{\lceil}{\rceil}
\DeclarePairedDelimiter{\floor}{\lfloor}{\rfloor}

\let \VEC \mathbf
\let \vec \mathbf

\long\def\comment#1{}

\let \lessthan \prec
\let \morethan \succ

\newcounter{l1}
\newcounter{l2}
\newcounter{l3}
\newcommand{\bdotlist}{\begin{list}{$\bullet$}{}}
\newcommand{\bboxlist}{\begin{list}{$\Box$}{}}
\newcommand{\bbboxlist}{\begin{list}{\raisebox{.005in}{{\tiny $\blacksquare$ \ \ }}}{}}
\newcommand{\bdashlist}{\begin{list}{$-$}{} }
\newcommand{\blist}{\begin{list}{}{} }
\newcommand{\barablist}{\begin{list}{\arabic{l1}}{\usecounter{l1}}}
\newcommand{\balphlist}{\begin{list}{(\alph{l2})}{\usecounter{l2}}}
\newcommand{\bAlphlist}{\begin{list}{\Alph{l2}.}{\usecounter{l2}}}
\newcommand{\bdiamlist}{\begin{list}{$\diamond$}{}}
\newcommand{\bromalist}{\begin{list}{(\roman{l3})}{\usecounter{l3}}}

\usepackage{tikz}
\usepackage{rotating}
\usepackage{pgfplots}
\pgfplotsset{compat=newest}
\pgfplotsset{plot coordinates/math parser=false}

\begin{document}

\title{
Rate-constrained Energy Services: \\ Allocation Policies and Market Decisions$^\pi$
}

\author{Ashutosh Nayyar$^a$,  Matias Negrete-Pincetic$^b$, Kameshwar Poolla$^c$, Pravin Varaiya$^c$
\thanks{$^a$Ming Hsieh Department of Electrical Engineering, University of Southern California. Corresponding author. ashutosn@usc.edu.}
\thanks{$^b$Department of Electrical Engineering of Pontificia Universidad Catolica de Chile, Santiago, Chile.}
\thanks{$^c$Electrical Engineering \& Computer Sciences, University of California, Berkeley, CA 94720.}
\thanks{$^\pi$Supported in part by EPRI and CERTS under sub-award 09-206; PSERC S-52; NSF under Grants 1135872, EECS-1129061, CPS-1239178, and CNS-1239274; the Republic of Singapore's National Research Foundation through a grant to the Berkeley Education Alliance for Research in Singapore for the SinBerBEST Program.}
}

\maketitle
\thispagestyle{empty}

\begin{abstract}


The integration of renewable generation poses operational and economic challenges for the electricity grid. 
For the core problem of power balance, the legacy paradigm of tailoring supply to follow random demand may be inappropriate 
under deep penetration of uncertain and intermittent renewable generation. In this situation, there is an emerging consensus that the 
alternative approach of controlling demand to follow random supply offers compelling economic benefits in terms of reduced regulation costs. 
This approach exploits the flexibility of demand side resources and requires sensing, actuation, and communication infrastructure; 
distributed control algorithms; and viable schemes to compensate participating loads.
This paper considers {\em rate-constrained energy services} which are a specific paradigm for flexible demand. 
These services are characterized by a specified delivery window, the total amount of energy that must be supplied over this window,
and the maximum rate at which this energy may be delivered.  We consider a forward market where rate-constrained energy services are traded.   
We explore allocation policies and market decisions of a supplier in this market.  The supplier owns a generation mix that includes some uncertain renewable generation and may also purchase energy in day-ahead and real-time markets to meet customer demand. 
The supplier must optimally select  the portfolio of rate-constrained services to sell, the amount of day-ahead energy to buy, 
and the policies for making real-time energy purchases and allocations to customers  to maximize its expected profit.  We offer solutions to the supplier's decision and control problems to economically provide rate constrained energy services.

\end{abstract}

\section{Introduction} \label{sec-introduction}
The worldwide interest in renewable energy is driven by pressing environmental problems, energy supply security and nuclear power safety concerns. 
The energy production from these renewable sources is {\em variable}: uncontrollable, intermittent, uncertain. 
Variability is the core challenge to deep renewable integration.

A central problem in this context is that of economically balancing demand and supply of electricity in the 
presence of large amounts of variable generation. The legacy {\em supply side} approach 
is to absorb the  variability in operating reserves.
Here, renewables are treated as negative demand, so the variability appears as uncertainty in net load. 
This variability is handled by scheduling  fast-acting reserve generation.  This strategy of {\em tailoring supply to meet demand} works
at today's modest penetration levels. But it will not scale. Recent studies in California \cite{CAISO}  project that the load-following capacity requirements will need to increase from 2.3 GW to 4.4 GW. These large increases in reserves will significantly raise electricity cost, and diminish the net carbon benefit from renewables (\cite{kirschen2010,negwankowshamey12}).

There is an emerging consensus that demand side resources must play a key role in supplying zero-emissions regulation services 
that are necessary for deep renewable integration  (e.g., \cite{cal09K, galus2010, papaoren2010, matdyscal12K, anand2012}). These include thermostatically controlled loads (TCLs), electric vehicles (EVs), and smart appliances.  Some of these loads are deferrable: they can be shifted over time.  For example, charging of electrical vehicles (EVs) may be postponed to some degree.    Other loads such as HVAC units can be modulated within limits.  
The core idea of {\em demand side approaches} to renewable integration is to exploit load flexibility to 
track variability in supply, i.e., to tailor demand to match supply.  For this an aggregator 
offers a control and business interface between the flexible loads and the system operator (SO). 

The demand side approach has led to two streams of work: (a) indirect load control (ILC) where flexible loads respond, in real-time, to price proxy signals, and (b) direct load control (DLC) where flexible loads cede physical control of devices to operators (cluster managers) who determine appropriate actions. The advantage of DLC is that with greater control authority the cluster manager can more reliably control the aggregate load.  
However, DLC requires a more extensive control and communication infrastructure and the manager must provide economic incentives to recruit a sufficient consumer base. The advantage of ILC is that the consumer retains authority over her electricity consumption.  

Current research in { direct load control}  focuses on developing and analyzing algorithms for coordinating resources  \cite{galus2010, papaoren2010}, \cite{ChenLLF2011, LeeTPS2008, HsuDLC1991, MetsEVCharging2010, GanPESGM2012, MaIFAC2011, GHGIREP2010}. For example, \cite{GanPESGM2012} develops a distributed scheduling protocol for electric vehicle charging; \cite{papaoren2010} uses approximate dynamic programming to couple wind generation with deferrable loads; and \cite{MaIFAC2011, GHGIREP2010, galus2010} suggest the use of receding horizon control approaches for resource scheduling.

Recent studies in { indirect load control} have developed real-time pricing algorithms \cite{IlicTPSNov2011,Paschalidis} and
quantified operational benefits \cite{LijesenRTPElasticity}.  There has also been research focused on economic efficiency in \cite{Borenstein2005,Spees2007}, feedback stability of price signals in \cite{RoozbehaniCDC2010}, volatility of real-time markets in \cite{wannegkowshameysha11b} as well as the practical issues associated with implementing ILC programs presented in \cite{LBNL-RTPReport}. 

Both ILC and DLC require appropriate economic incentives for the consumers.  In ILC, the real-time price signals provides the required incentives. However, the quantification of those prices, the feasibility of consumer response and the impact on the system and market operations in terms of price volatility and instabilities is a matter of concern, as recent literature suggests (\cite{RoozbehaniCDC2010,LBNL-RTPReport,wannegkowshameysha11b}).  

DLC also requires the creation of economic signals, but unlike real-time pricing schemes DLC can  use forward markets. For DLC to be effective, it is necessary to offer consumers who present greater demand flexibility a larger discount. The discounted pricing can be arranged through flexibility-differentiated electricity markets. Here, electricity is regarded as a set of differentiated services as opposed to a homogeneous 
commodity.  Consumers can purchase an appropriate bundle of services that best meets their electricity needs.
From the producer's perspective, providing differentiated services may better accommodate supply variability. There is a growing body of work \cite{tanvar93, pravin2011,negmey12, bitar2013} on differentiated electricity services. 
An early exposition of differentiated energy services is offered in \cite{OrenSmith}.
There are other approaches to such services that naturally serve to integrate variable generation sources. 
Reliability differentiated energy services where consumers accept contracts for $p$ MW of power with probability $\rho$ are developed in \cite{tanvar93}. More recently, the works of \cite{bitarlow2012, bitar2013} consider deadline differentiated contracts where consumers receive price discounts for offering larger windows for the delivery of $E$ MWh of energy. 

This paper follows the direct load control paradigm. We consider a supplier proving rate-constrained energy series to its customers.  These services are characterized by the total amount of energy that must be delivered over a specified time period and the maximum rate at which this energy may be delivered. Examples of loads that can use these services include electric vehicles, pumping systems and smart appliances. The supplier can use its uncertain renewable generation as well as energy purchased in day-ahead and real-time energy markets to serve its customers. The supplier has to optimize  the portfolio of rate-constrained services to sell, the amount of day-ahead energy to buy and the policies for making real-time energy purchases and allocations to customers in order to maximize its expected profit. We provide solutions to the supplier's decision and control problems in this paper.

This paper is organized as follows.
We describe the supplier's operation and formulate its optimization problem in Section~\ref{sec:PF}. For a fixed choice of forward market decisions, we find optimal real-time energy purchase and energy allocation policies for the supplier in Section~\ref{sec:method}. Using these optimal policies, we provide a solution strategy for the supplier's optimization problem in Section~\ref{sec:profit}. We draw conclusions in Section~\ref{sec-conclusions}. 

\subsection*{Notation}
 Capital letters are random variables/random vectors. Small letters are realizations. Bold letters denote vectors. 
For a vector $\VEC a =(a_1,a_2,\ldots,a_T)$, $\VEC a^{\downarrow} =(a^{\downarrow}_1,\ldots, a^{\downarrow}_T)$ denotes the non-increasing 
rearrangement of $\VEC a$, so, $a^{\downarrow}_t \geq a^{\downarrow}_{t+1}$ for $t=1,2,\ldots,T-1$.
For an assertion $A$, $\mathds{1}_{A}$ denotes $1$ if $A$ is true and $0$ if $A$ is false. 
$(x)^+$ denotes $\max(x,0)$, $\floor{x}$ denotes the largest integer less than or equal to $x$, and $\ceil{x}$ denotes the smallest integer greater than or equal to $x$.

\section{Problem Formulation}\label{sec:PF}
 We consider a supplier providing rate-constrained energy services to consumers. The services are to be delivered over a time interval of $T$ units divided into discrete slots of unit duration. A rate-constrained energy service is characterized by the pair $(E,m)$, where $E$ represents the total amount of energy and $m$ represents the maximum energy that can be delivered per time slot. $m$ is referred to as the rate-constraint of the service.   The pair $(E,m)$ is restricted to be within a  bounded set in the non-negative quadrant of $\mathbb{R}^2$.  Services must also satisfy a basic feasibility condition: $E \leq mT$.
 
  A $(E,m)$ service can be provided by any energy allocation $u_1,u_2,\ldots,u_T$ that satisfies
 \[ \sum_{t=1}^T u_t = E, ~~~ 0 \leq u_t \leq m. \]
 The supplier sells the services to consumers in a forward market. The market specifies a price $\pi(E,m)$ for the service $(E,m)$. The market prices are known a priori and the supplier  and consumers act as price-takers in the forward market for services.
 
\subsection{Energy sources}
The supplier can use three types of energy sources to provide the services:
\begin{enumerate}
\item \emph{Renewable energy:} The supplier has access to an uncertain supply of renewable energy available at zero marginal cost. The renewable supply over the delivery period is denoted by the random vector $\vec{R} = (R_1,R_2,\ldots,R_T)$, where $R_t$ is the energy available in the $t$-th time slot. Its realizations are denoted by $\vec{r} =(r_1,r_2,\ldots,r_T)$.

\item  \emph{Day-ahead energy:} The supplier can purchase energy in a day-ahead market at a price $c^{da}$. The energy  purchased in the day-ahead market is used  over the delivery period. The day-ahead energy is denoted by $\vec{y} = (y_1,\ldots,y_T)$, where $y_t$ is energy available in the $t$-th time slot.

\item \emph{Real-time energy:} In addition to the day-ahead purchases, the supplier can purchase energy in a real-time market at the beginning of each of the $T$ time slots. The real-time energy price is $c^{rt}$. 
\end{enumerate}  

\subsection{Supplier's operation}
Suppose the supplier sells $n$ services  in the forward market for services. We will denote by $\mathcal{C}$ the set of services sold. Thus,
\[ \mathcal{C} = \{(E^1,m^1),(E^2,m^2),\ldots,(E^n,m^n) \} \]
 The $i^{th}$ service is delivered to the $i^{th}$ consumer.
 The supplier  earns a total revenue of
 \[ \sum_{i=1}^n \pi(E^i,m^i).\]
 At the time of making its day-ahead purchase, the supplier has a forecast model for its renewable supply specified as a probability distribution for the random vector $\vec{R}$. Suppose the supplier purchases the supply vector $\vec{y}$ in the day-ahead energy market. Then, its day-ahead cost is 
 \[ c^{da}\sum_{t=1}^T y_t.\]
 The real-time operation of the supplier is as follows:
 \begin{enumerate}
 \item[(i)] At the beginning of the $t$-th time slot, the supplier observes the renewable supply, $R_t$, for that slot. 
 \item[(ii)] It decides the amount of real-time energy to purchase in that slot. This amount is denoted by $A_t$ and is chosen according to a real-time purchase policy $g=(g_1,\ldots g_T)$ so that
 \[ A_t = g_t(\mathcal{C},\vec{y},R_1,R_2,\ldots,R_t).\]
 \item[(iii)] The total energy  available in the $t$-th slot is $y_t +R_t +A_t$. The supplier allocates this energy among its consumers. The energy allocated to the the $i$th consumer is denoted by $U^i_t$. The allocation amounts are chosen according to an allocation policy $h^i = (h^i_1,\ldots,h^i_T), i=1,\ldots,n$  so that
 \[ U^i_t =h^i_t(\mathcal{C},\vec{y},R_1,R_2,\ldots,R_t),\]
 with the constraint that total allocated energy $\sum_i U^i_t$ in the $t$th slot cannot exceed the total available energy in that slot, that is,
 \[ \sum_{i=1}^n U^i_t \leq y_t +R_t+A_t.\]
 \end{enumerate}
 The supplier is required to fully satisfy its consumers. That is, the $i$th consumer who purchased  the service $(E^i,m^i)$ must receive an allocation that satisfies
 \[ \sum_t U^i_t = E^i, ~~ 0 \leq U^i_t \leq m^i.\]
 
 \subsection{Discretization of services}
The pair $(E,m)$ characterizing a service can take any non-negative  value from a bounded set in $\mathbb{R}^2$. This leads to a continuum of possible services in the market. For both mathematical and modeling reasons, it will be convenient to assume that energy is sold and allocated in discrete multiples of some basic unit (say, $kW\times\mbox{duration of one time slot}$).  This implies  that $E$, $m$  as well as the allocations $U^i_t$ take discrete (non-negative integer) values. In particular, the allocation $U^i_t$ in the supplier's optimization problem is constrained to take values in the finite set $\{0,1,\ldots,m^i\}$.
\subsection{Optimization problem}\label{sec:profit_prob}
The supplier has to select the set $\mathcal{C}$ of services to be sold (both the cardinality $n$ and the types of services) in the forward market, the day-ahead energy to purchase, a real-time purchase policy $g$ and allocation policies $h^i, i=1,2,\ldots,n$ in order to maximize its expected profit
\[ J := \sum_{i=1}^n \pi(E^i,m^i) - c^{da}\sum_{t=1}^T y_t - c^{rt}\mathds{E}\left[\sum_{t=1}^T A_t\right] \]
subject to 
\[ \sum_{t=1}^T U^i_t = E^i, ~~ U^i_t \in \{0,1,\ldots, m^i\}, ~i=1,\ldots,n,\]
\[ \sum_{i=1}^n U^i_t \leq y_t +R_t+A_t,\]
  \[A_t = g_t(\mathcal{C},\vec{y},R_1,R_2,\ldots,R_t), \]
  and \[U^i_t =h^i_t(\mathcal{C},\vec{y},R_1,R_2,\ldots,R_t),~i=1,\ldots,n.\]
  
  \begin{remark}
  Note that $c^{da},c^{rt}$ are the day-ahead and real-time prices, respectively, for $1$ (discrete) unit of energy. These prices would be scaled appropriately if the discretization unit is changed.
  \end{remark}
 \section{Optimal policies for fixed $\mathcal{C}, \vec{y}$}\label{sec:method}
For a fixed choice of $\mathcal{C}, \vec{y}$, the maximum achievable profit can be written as 
\begin{equation} \label{eq:profit}
 J(\mathcal{C},\vec{y}) =  \sum_{i=1}^n \pi(E^i,m^i) - c^{da}\sum_{t=1}^T y_t - V(\mathcal{C},\vec{y}), 
\end{equation}
where 
\[ V(\mathcal{C},\vec{y}) = \min_{g,h} c^{rt}\mathds{E}\left[\sum_{t=1}^T A_t\right]\]
subject to
\begin{equation}\label{eq:constraints}
\begin{aligned} 
 &\sum_t U^i_t = E^i, ~~  U^i_t \in \{0,1,\ldots, m^i\}, ~i=1,\ldots,n,  \\
 &~~~~~\sum_{i=1}^n U^i_t \leq y_t +R_t+A_t,\\
 &A_t = g_t(\mathcal{C},\vec{y},R_1,R_2,\ldots,R_t),\\
 &U^i_t =h^i_t(\mathcal{C},\vec{y},R_1,R_2,\ldots,R_t),~i=1,\ldots,n.
\end{aligned}
\end{equation}
 
 Our main goal in this section is to characterize the function $V(\mathcal{C},\vec{y})$. Our approach will be to first find a lower bound on $V(\mathcal{C},\vec{y})$ under a \emph{information relaxation} assumption  where a oracle has revealed a priori the realization of renewable generation, and then constructing real-time policies that achieve this lower bound.  We will proceed in the following steps:
\begin{enumerate}
\item[1.]  Suppose it was known a priori (that is, before the start of real-time operation) that the  total available supply  is going to be $\vec{p} = (p_1,p_2,\ldots,p_T)$. We will find conditions under which this supply by itself is adequate for providing all the services in $\mathcal{C}$.
\item[2.] Suppose the realization of the renewable supply was known a priori to be $\vec{r} = (r_1,r_2,\ldots,r_T)$. Then, without any real-time purchases, the available supply would be $\vec{p}=\vec{y}+\vec{r}$. Using the conditions obtained in Step 1, we can determine if  the supply $\vec{y} + \vec{r}$ is adequate for providing all the services in  $\mathcal{C}$. If the supply is not adequate, we will find the minimum amount of additional energy required to achieve adequacy. Let this amount be denoted by $L(\mathcal{C}, \vec{y},\vec{r})$.
\item[3.] Recall that under a real-time purchase policy, the purchase decision at time slot $t$ is made without knowing the realization of future renewable generation.  We will argue that in order to satisfy all the problem constraints of \eqref{eq:constraints}, a real-time purchase policy must satisfy the inequality
\begin{equation}\label{eq:lowerbound}
\sum_{t=1}^T g_t(\mathcal{C},\vec{y},r_1,r_2,\ldots,r_t) \geq L(\mathcal{C}, \vec{y},\vec{r}), 
\end{equation}
for each realization $\vec{r}=(r_1,\ldots,r_T)$, where $L(\mathcal{C}, \vec{y},\vec{r})$ is the lower bound obtained in Step 2.
In other words, for each realization $\vec{r}$, a real-time purchase policy $g$ must purchase at least $L(\mathcal{C}, \vec{y},\vec{r})$ amount of energy in order to provide all the services.
\item[4.] Finally, we will construct a real-time purchase policy $g^*$ that achieves the lower bound  in \eqref{eq:lowerbound}. This would imply that the expected cost of real-time purchases under this policy is equal to $V(\mathcal{C},\vec{y})$.
\end{enumerate}

\subsection{Step 1: Characterizing Adequacy}
Consider  a given set of services $\mathcal{C}$. Suppose that the available supply is  known a priori to be $\vec{p}$. We want to determine if $\vec{p}$
 is adequate for the set of services $\mathcal{C}$ without any additional energy purchases. We can formally define the adequacy of $\vec{p}$ for $\mathcal{C}$ as follows.
 \begin{definition}
  The available supply $\vec{p}$ is adequate for the services in $\mathcal{C}$ if there exist allocations $u^i_t, i=1,2,\ldots,n$, $t=1,2,\ldots,T$ that satisfy 
\[ \sum_t u^i_t = E^i, ~~  u^i_t \in \{0,1,\ldots, m^i\}, ~i=1,\ldots,n,\]
\[ \sum_{i=1}^n u^i_t \leq p_t.\]
\end{definition}

To answer the question of adequacy, we will take a closer look at the allocations needed for each individual service. Consider the service $(E,m)$. Since we assumed discrete values for $E$ and $m$, we can find non-negative integers $k, \ell$ with $\ell < m$ such that 
\[ E =km+\ell.\]
We now define a decomposition of the service $(E,m)$ into $m$ separate services all of  which have a rate constraint  equal to $1$.

\begin{definition}
For a service $(E,m)$ with  $E = km + \ell$, $\ell< m$, we define a unit-rate decomposition of the service, denoted by $\mathcal{D}(E,m)$, as the following collection of services
\begin{align}
\mathcal{D}(E,m) := \{ (E^j, 1)| j =1,2,\ldots,m\},
\end{align}
where 
\begin{align}
E^j = \left\{ \begin{array}{ll}
k+1 ~\mbox{for $j \leq \ell$,}\\
k~~~~ ~\mbox{for $j > \ell$}
\end{array}
\right..
\end{align}
\end{definition}

The following lemma provides an equivalence between the service $(E,m)$ and its unit rate decomposition $\mathcal{D}(E,m)$ from an allocation perspective.

\begin{lemma} \label{lemma:unitrate}
For a service $(E,m)$ with  $E = km + \ell$, $\ell< m$, an allocation $u_t,t=1,\ldots,T$ satisfies
\[ \sum_{t=1}^T u_t = E, ~~  u_t \in \{0,1,\ldots, m\}\]
if and only if $u_t = \sum_{j=1}^{m} \nu^j_t$, $t =1,\ldots, T$ such that 
\begin{align}
& \sum_{t=1}^T \nu^j_t = k+1, ~\nu^j_t \in \{0,1\} ~~\mbox{for $j=1,2,\ldots,\ell^i$}, \notag \\
 & \sum_{t=1}^T \nu^j_t = k, ~\nu^j_t \in \{0,1\} ~~\mbox{for $j=\ell^i+1,\ldots, m^i$}
\end{align}
\end{lemma}
\begin{proof}
See Appendix \ref{sec:unitrate}.
\end{proof}

We define a unit-rate decomposition for a collection of services $\mathcal{C}$ as 
\begin{equation}
\mathcal{D}(\mathcal{C}) := \bigcup_{i=1}^n \mathcal{D}(E^i,m^i).
\end{equation}

The following result is a direct consequence of Lemma \ref{lemma:unitrate}.
\begin{lemma}\label{lemma:unitrate2}
The available supply $\vec{p}$ is adequate for set of services $\mathcal{C}$ if and only if it is adequate for $\mathcal{D}(\mathcal{C})$. 
\end{lemma} 

The utility of Lemma \ref{lemma:unitrate2} comes from the fact that all services in $\mathcal{D}(\mathcal{C})$ have the same rate constraint, $m=1$. The services in $\mathcal{D}(\mathcal{C})$ differ only in the amount of total energy requested. Therefore, we can write this collection of services  as 
\[ \mathcal{D}(\mathcal{C}) = \{ (E^j,1)| j=1,2,\ldots, N\}.\]
We will focus on the adequacy characterization for $\mathcal{D}(\mathcal{C})$. We first define a stronger version of adequacy that removes the possibility of excess energy.
\begin{definition}
  The available supply $\vec{p}$ is exactly adequate for the services in $\mathcal{D}(\mathcal{C})$ if there exist allocations $\nu^j_t, i=1,2,\ldots,N$, $t=1,2,\ldots,T$ that satisfy 
\[ \sum_{t=1}^T \nu^j_t = E^j, ~~  \nu^j_t \in \{0,1\}, ~j=1,\ldots,N,\]
\[ \sum_{j=1}^N \nu^j_t = p_t.\]
\end{definition}

The binary allocation decisions $\nu^j_t$ can now be interpreted as on-off decisions for supplying energy to the $j$th service in $\mathcal{D}(\mathcal{C})$ at time $t$. The constraint $\sum_{t=1}^T \nu^j_t = E^j$ can then be viewed as specifying the total duration for which energy supply for the $j$th service is turned on. We will therefore refer to $E^j$ as the duration requested by the $j$th service in $\mathcal{D}(\mathcal{C})$. 

We now define a \emph{demand-duration vector} for a collection of unit-rate services.
\begin{definition}\label{def:dvector}
For the collection of services $\mathcal{D}(\mathcal{C})$, let $d_t$ be number of services with requested duration $E^j \geq t$, that is,
\[ d_t := \sum_{j=1}^N \mathds{1}_{\{ E^j \geq t\}}.\]
The vector $\vec{d} = (d_1,d_2,\ldots,d_T)$ is called the demand-duration vector for $\mathcal{D}(\mathcal{C})$.
\end{definition}

We will characterize adequacy in terms of a majorization relation between the supply $\vec{p}$ and the demand-duration vector $\vec{d}$. We need the following definitions.

\begin{definition}[Majorization]{\rm \label{def:majorization}
Let $\VEC a = (a_1,\cdots,a_T)$ and $\VEC b = (b_1,\cdots,b_T)$ be two non-negative  vectors. Denote by $\VEC a^{\downarrow}, \VEC b^{\downarrow}$ the non-increasing rearrangements of $\VEC a$ and $\VEC b$ respectively. 
We say that $\VEC a$ majorizes $\VEC b$, written $\VEC a \lessthan \VEC b$, if 
\bromalist
\item $\sum_{s=t}^T a^{\downarrow}_s \leq \sum_{s=t}^T b^{\downarrow}_s$, for $t=1,2,\ldots,T$, and 
\item $\sum_{s=1}^T a^{\downarrow}_s= \sum_{s=1}^T b^{\downarrow}_s$.
\end{list}
If only the first condition holds, we say that $\VEC a$ weakly majorizes $\VEC b$, written 
$\VEC a \lessthan^{w} \VEC b$.}
\end{definition}

\begin{remark}{\rm The inequalities in our definition of majorization are reversed from standard use in majorization theory. 
This departure from convention allows us to write our adequacy conditions
as $\VEC d \lessthan \VEC p$ and $\VEC d \lessthan^w \VEC p$  which resemble the more familiar adequacy condition of demand being ``less than'' supply .}
\end{remark}

\begin{theorem}[\sc Adequacy]\label{thm:two}
Let $\vec{d}$ be the demand-duration vector for the collection of services $\mathcal{D}(\mathcal{C})$. Then, 
\balphlist
\item The supply $\vec{p}$  is exactly adequate for $\mathcal{D}(\mathcal{C})$ if and only if $\VEC d \lessthan \VEC p$.
\item The supply  $\VEC p$  is adequate for $\mathcal{D}(\mathcal{C})$ if and only if 
 $\VEC d \lessthan^{w} \VEC p$.
 \end{list}
\end{theorem}
\proof{Proof}
 See Appendix \ref{sec:two}.
\endproof
\begin{remark}
It follows from the definition of $\vec d$ that $d_t \geq d_{t+1}$. Thus, $\vec d = \vec d^{\downarrow}$.
\end{remark}
 Given an adequate supply $\vec{p}$ for $\mathcal{D}(\mathcal{C})$, the following result provides an algorithm that finds the  allocations $\nu^j_t$.

\begin{theorem}\label{thm:allocation}
Suppose $\vec{p}$ is adequate for 
\[  \mathcal{D}(\mathcal{C}) = \{ (E^j,1)| j=1,2,\ldots, N\}.\]
Construct $\nu^j_t$ sequentially as follows:
\bromalist
  \item At time $t=1$, $x^j_1 := T-E^j$. Let $\pi$ be a permutation of $\{1,\ldots N\}$ such that $\pi(j) <\pi(k)$ implies $x^j_1 \leq x^{k}_1$. Set
  \begin{align}
  \nu^j_1 = \left\{ \begin{array}{ll}
  1 & \mbox{$\pi(j) \leq p_1$ and $E^j >0$}, \\
  0 & \mbox{otherwise}
  \end{array}
  \right. \label{eq:allocation1}
  \end{align}
  
  \item At time $t$,  $x^j_t := T -t+1-\left(E^j -\sum_{s=1}^{t-1}\nu^j_s\right)$. Let $\pi$ be a permutation of $\{1,\ldots N\}$ such that $\pi(j) <\pi(k)$ implies $x^j_t \leq x^{k}_t$. Set
  \begin{align}
  \nu^j_t = \left\{ \begin{array}{ll}
  1 & \mbox{$\pi(j) \leq p_t$ and $(E^j-\sum_{s=1}^{t-1}\nu^j_s) >0$}, \\
  0 & \mbox{otherwise}
  \end{array}
  \right. \label{eq:allocationt}
  \end{align} 
   \end{list}
   Then,  \[ \sum_{t=1}^T \nu^j_t = E^j, \quad \sum_{j=1}^N \nu^j_t \leq p_t. \]
\end{theorem}

\proof
See Appendix \ref{sec:app_3}.
\endproof
 \begin{remark}
 The quantity $x^j_t$ can be interpreted as the \emph{laxity}  associated with service $j$ at time $t$. The allocation algorithm of Theorem \ref{thm:allocation} is therefore a least laxity first algorithm \cite{anand2012}.
 \end{remark}
Note that even though we assumed that the supply $\vec{p}$ is known in advance, the allocation algorithm described in Theorem \ref{thm:allocation} is actually \emph{causal}, that is, the allocations $\nu^j_t$ at time $t$ do not depend on future supply $p_{t+1},\ldots,p_T$.  This causality will be crucial for constructing the allocation policies $h^i$ described in Section \ref{sec:PF}.

\subsection{Step 2: Energy gap for adequacy}
Consider  a given set of services $\mathcal{C}$, day ahead energy supply $\vec{y}$. Assume that the renewable supply is known a priori to be $\vec{r}$ and let $\vec{p} = \vec{y} + \vec{r}$. The supply vector $\vec{p}$ may or may not be adequate for $\mathcal{D}(\mathcal{C})$. 

Define
\begin{equation} \label{eq:LBfunc}
L(\mathcal{C},\vec{y},\vec{r}) := \left[\max_{1 \leq t \leq T} \sum_{ s = t}^{T} (d_s -p^{\downarrow}_s)\right]^+
\end{equation}
where $\vec{p} = \vec{y} + \vec{r}$ and $\vec{d}$ is the demand duration vector associated with $\mathcal{D}(\mathcal{C})$. The following lemma establishes that $L(\mathcal{C},\vec{y},\vec{r})$ is a lower bound on the  amount of additional energy that must be added to $\vec{p}$ to achieve adequacy.

\begin{lemma} \label{lemma:LB}
Let $\vec{q} \geq \vec{p}$ (that is, $q_t \geq p_t$ for all $t$) be any vector that satisfies $\vec d \lessthan^w \vec q$. Then,
\begin{equation}
\sum_{t=1}^T (q_t-p_t) \geq  \left[\max_{1 \leq r \leq T} \sum_{ s = r}^{T} (d_s -p^{\downarrow}_s)\right]^+
\end{equation}
\end{lemma}
\begin{proof}
See Appendix \ref{sec:LB}.
\end{proof}
\begin{remark}
The vector $\vec p^{\downarrow}$ is referred to as the supply-duration vector.
\end{remark}
\subsection{Step 3: Lower bound on the cost of real-time energy}
We now return to the problem of finding real-time purchase and allocation policies $g,h^i,i=1,\ldots,n$. Consider any choice of policies  that satisfy the constraints in \eqref{eq:constraints} for all realizations of renewable supply. Consider a realization $\vec{r}$ of the renewable supply. Then, the total supply under the real-time purchase policy $g$ is 
\begin{equation}
q_t = y_t+r_t+g_t(\mathcal{C},\vec{y},r_1,\ldots,r_{t}), ~~ t=1,\ldots,T.
\end{equation}
Let $u^i_t$ be the realizations of the allocation decisions under the allocation policy $h^i,i=1,\ldots,n$. Since we know that the policies $g,h^i,i=1,\ldots,n$ satisfy the constraints \eqref{eq:constraints} for all realizations of renewable supply, it follows that 
\begin{equation}
\begin{aligned} 
 &\sum_t u^i_t = E^i, ~~  u^i_t \in \{0,1,\ldots, m^i\}, ~i=1,\ldots,n,  \\
 &~~~~~\sum_{i=1}^n u^i_t \leq q_t =y_t +r_t+g_t(\mathcal{C},\vec{y},r_1,\ldots,r_{t}),
\end{aligned}
\end{equation}
Thus the supply $\vec{q} = (q_1,\ldots,q_T)$ is adequate for $\mathcal{C}$. By Lemma \ref{lemma:LB}, it follows that 
\begin{equation}
\sum_{t=1}^T (q_t-p_t) \geq  \left[\max_{1 \leq r \leq T} \sum_{ s = r}^{T} (d_s -p^{\downarrow}_s)\right]^+, \label{eq:step3eq1}
\end{equation}
where $p_t=y_t+r_t$. \eqref{eq:step3eq1} implies that 
\begin{equation}
\sum_{t=1}^T g_t(\mathcal{C},\vec{y},r_1,\ldots,r_{t})\geq  \left[\max_{1 \leq r \leq T} \sum_{ s = r}^{T} (d_s -p^{\downarrow}_s)\right]^+.
\end{equation}
In other words, for each realization $\vec{r}$, the function $L(\mathcal{C},\vec{y},\vec{r})$ is a lower bound on the total amount of  real-time energy purchased.

\subsection{Step 4: Construction of optimal policies}

Consider the collection of services $\mathcal{C}$ with the associated decomposition $\mathcal{D}(\mathcal{C})$ and  the demand-duration vector $\vec{d}$ and the day-ahead energy $\vec{y}$. Consider the following real-time purchase policy, $g^*$:
\bromalist
\item  At $t=1$, for a realization $r_1$ of $R_1$, the real-time purchase is $a_1 = g^*_1(\mathcal{C},\vec y, r_1) :=  (d_T - y_1-r_1)^+ $.
\item At time $t$, for a realization $r_1,\ldots,r_t$ of renewable supply and  the realization $a_1,\ldots, a_{t-1}$ of past real-time purchases, $a_t = g^*_t(\mathcal{C},\vec y,r_1,\ldots,r_t )$ is defined as the solution of the following minimization problem
\begin{align} \nonumber
 &\min_{a_t \geq 0} a_t  \\
 s.t.~~&(y_1+r_1+a_1,\ldots,y_{t-1}+r_{t-1}+a_{t-1},y_t+r_t+a_t)  \notag \\&~~~~~ \morethan_w (d_{T-t+1},d_{T-t+2},\ldots,d_T), \label{eq:rt_purch}
\end{align}
where, for $s<t$, $a_s = g^*_s(\mathcal{C},\vec y,r_1,\ldots,r_s)$.
\end{list}
Note that the majorization relation in \eqref{eq:rt_purch} is between two $t-$ dimensional vectors.
\begin{theorem}\label{thm:main_thm}
\begin{enumerate}
\item For every realization $\vec{r}$ of renewable supply, the policy $g^*$ generates a total supply vector $\vec{q} = (q_1,\ldots,q_T)$, with $q_t = y_t +r_t + g^*_t(\mathcal{C}, \vec y, r_1,\ldots,r_t)$, that satisfies
\[ \vec{d} \lessthan^w \vec q\]
\item Consider the virtual set of services $\mathcal{D}(\mathcal{C})$. The allocation decisions $\nu^j_t$ for the virtual services  in $\mathcal{D}(\mathcal{C})$ are obtained using the (causal) allocation algorithm of Theorem \ref{thm:allocation} by replacing $p_t$ with $q_t$.  For a service $(E^i,m^i) $ in the set of actual services $\mathcal{C}$,  
\[ \mathcal{D}(E^i,m^i) = \{ (E^{j_1},1), \ldots, (E^{j_{m^i}},1) \} \subset \mathcal{D}(\mathcal{C}). \]
The actual allocation decision for this service is $u^i_t = \sum_{k=1}^{m^i} \nu^{j_k}_t$.
The allocations satisfy the constraints \eqref{eq:constraints} for each realization of $\vec{r}$.

\item For each $\vec{r}$, the total amount of real-time purchases under $g^*$ is equal to the lower bound $L(\mathcal{C},\vec{y},\vec{r})$. Thus, $g^*$ is optimal and 
\begin{align}
V(\mathcal{C},\vec{y}) &= \mathds{E}_{\vec{R}} [L(\mathcal{C},\vec{y},\vec{R})] \notag \\
&= \mathds{E}_{\vec{R}}\left[\max_{1 \leq t \leq T} \sum_{ s = t}^{T} (d_s -P^{\downarrow}_s)\right]^+,
\end{align}
where $\vec{P} = \vec{y} + \vec{R}$.
\end{enumerate}
\end{theorem}
\begin{proof} See Appendix \ref{sec:main_thm}. \end{proof}
\begin{remark}
The above results are derived under the assumption that real-time energy price is constant over the time horizon under consideration. It remains to be seen if analogous results can be obtained with time-varying prices.  
\end{remark}
\section{Profit Maximization}\label{sec:profit}
We now return to the expected profit maximization problem formulated in Section \ref{sec:profit_prob}. Using \eqref{eq:profit} from Section \ref{sec:method}, the expected profit maximization problem can be written as
\begin{align}
J^*&= \max_{\mathcal{C},\vec y} J(\mathcal{C},\vec{y}) \notag \\
&= \max_{\mathcal{C},\vec y} \left[  \sum_{i=1}^n \pi(E^i,m^i) - c^{da}\sum_{t=1}^T y_t - V(\mathcal{C},\vec{y})\right],
\end{align}
which by Theorem \ref{thm:main_thm} can be written as
\begin{align}
 \max_{\mathcal{C},\vec y}   \Bigg[&\sum_{i=1}^n \pi(E^i,m^i) - c^{da}\sum_{t=1}^T y_t \notag \\ &- \mathds{E}_{\vec{R}}\left[\max_{1 \leq t \leq T} \sum_{ s = t}^{T} (d_s -P^{\downarrow}_s)\right]^+\Bigg], \label{eq:max_prob}
\end{align}
where $\vec{P} = \vec{y} + \vec{R}$, and $\vec d$ is the demand-duration vector associated with $\mathcal{D}(\mathcal{C})$.

\emph{Market Price Assumption:} Consider a service $(E,m)$ with its unit rate decomposition being
\[ \mathcal{D}(E,m) = \{ (E^1,1), \ldots, (E^m,1)\}. \]
We will assume that the market price for the service $(E,m)$ is equal to the sum of prices for services in its unit rate decomposition. That is,
\begin{equation}
\pi(E,m) = \sum_{j=1}^m \pi(E^j,1) \label{eq:price_assum}
\end{equation}
This assumption can be justified by the following reasoning: Suppose $\pi(E,m) >\sum_{j=1}^m \pi(E^j,1)$. Then, because of Lemma \ref{lemma:unitrate}, a consumer who wants to purchase the service $(E,m)$ can get the same allocation at a cheaper price  by purchasing the collection of services $\mathcal{D}(E,m)$ instead. Thus, effectively, an $(E,m)$ service can be purchased at  the price of $\sum_{j=1}^m \pi(E^j,1)$ in this market.
Conversely, suppose that $\pi(E,m) < \sum_{j=1}^m \pi(E^j,1)$. Because of the argument of Section \ref{sec:method}, $V(\mathcal{C},\vec y) = V(\mathcal{D}(\mathcal{C}), \vec y)$. Therefore, if $\pi(E,m) < \sum_{j=1}^m \pi(E^j,1)$, then it follows that $J(\mathcal{C},\vec y) < J(\mathcal{D}(\mathcal{C}), \vec y)$. Thus, the supplier would only sell unit rate services and a consumer requiring $(E,m)$ service would have to purchase the collection $\mathcal{D}(E,m)$ instead. The effective market price for $(E,m)$ service would once again be  $\sum_{j=1}^m \pi(E^j,1)$\footnote{We are assuming that the bundling of services incurs no additional costs.}.

Because of \eqref{eq:price_assum}, $J(\mathcal{C},\vec y) = J(\mathcal{D}(\mathcal{C}), \vec y)$. Therefore, we will consider only collections of unit rate services for profit maximization. The optimal collection of unit rate services can then be bundled and sold in different combinations without changing the profit. The revenue from selling a collection of unit rate services $\{(E^j,1)|j=1,\ldots,N\}$ is 
\begin{align}
&\sum_{j=1}^N \pi(E^j,1) = \sum_{t=1}^T \pi(t,1)\left(\sum_{j=1}^N \mathds{1}_{\{ E^j = t\}}\right) \notag \\
&= \sum_{t=1}^T \pi(t,1) \left( d_t - d_{t+1}\right), \label{eq:new_rev}
\end{align}
where we used the definition of demand duration vector $\vec d$ (Definition \ref{def:dvector}) with $d_{T+1}:=0$. Using \eqref{eq:new_rev}, the maximization problem of \eqref{eq:max_prob} can be written as
\begin{align}
& \max_{\vec d,\vec y} J(\vec d,\vec y) \notag \\
 &:= \max_{\vec d,\vec y} \Bigg[ \sum_{t=1}^T \pi(t,1) \left( d_t - d_{t+1}\right)  - c^{da}\sum_{t=1}^T y_t \notag \\
 &~~- \mathds{E}_{\vec{R}}\left[\max_{1 \leq t \leq T} \sum_{ s = t}^{T} (d_s -P^{\downarrow}_s)\right]^+ \Bigg], \label{eq:max_prob2}
\end{align}
where $\vec{P} = \vec{y} + \vec{R}$. The optimization problem in \eqref{eq:max_prob2} is an integer programming problem since $\vec {d}, \vec y$ are constrained to be non-negative integers (recall that the services and allocations were assumed to be discrete-valued). Approximate solutions to this integer program can be obtained by first solving the optimization problem without integer constraints and then rounding up/down the non-integer solution. The following result establishes the efficiency and accuracy of this approach.

\begin{theorem}\label{thm:opt}
1. The optimization problem of \eqref{eq:max_prob2} without the integer constraints is a convex optimization problem. We denote its solution by $(\vec{d}^c, \vec{y}^c)$. \\
2. Construct an  integer solution $(\vec d^a, \vec y^a)$ from $(\vec{d}^c, \vec{y}^c)$ as follows: 
\begin{align}
(d^a_t - d^a_{t+1}) &= \floor{d^c_t-d^c_{t+1}} \\
y^a_t &= \ceil{y^c_t} \label{eq:approx}
\end{align}

 Let $\vec{d}^*,\vec{y}^*$ be the maximizing vectors for the integer program of \eqref{eq:max_prob2}. Then,
 \begin{equation} \label{eq:gap}
 J(\vec d^*, \vec y^*) - J(\vec d^a, \vec y^a) \leq c^{da}T + \sum_{t=1}^T \pi(t,1).
 \end{equation}
\end{theorem}
\begin{proof}
See Appendix \ref{sec:opt}.
\end{proof}
\begin{remark}
Note that the price of unit rate services and the day-ahead  energy price per unit will be lower if the discretization unit for energy is smaller. Thus, the optimality gap in \eqref{eq:gap} can be reduced by adopting a finer discretization of energy. In \cite{nayyar_tcns}, we consider a situation where the discretization is taken to the limit and power levels can take real-number values.
\end{remark}

\section{Conclusions and Future Work} \label{sec-conclusions}

Flexible loads are expected to play a central role in supporting deep renewable integration.
They enable demand shaping to balance supply variability, and thus offer an effective alternative to 
conventional generation reserves. In this paper, we studied a stylized model of flexible loads that require rate-constrained energy services. Such a service is characterized by the total amount of energy that must be delivered over a specified time period and the maximum rate at which this energy may be delivered. We addressed the problem of finding optimal market decisions and allocation policies for a supplier that  can sell  rate-constrained energy services in a forward market and can use its uncertain renewable generation as well as energy purchased in day-ahead and real-time energy markets to serve its customers.

The theoretical analysis of this work can facilitate further studies required for exploiting load flexibility. For example,  the results of this paper can be used to evaluate the value of flexibility in reducing supplier's cost under different forecast models of renewable generation. Further, since supply adequacy is not assessed at each time instant, forecasting the exact profile of renewable generation may not be necessary.  Instead, the results of this paper suggest that for providing rate-constrained energy services, the key metric to be estimated is the supply-duration vector of the renewable generation. Other directions for future work include investigation of competitive equilibria in the forward market for energy services and comparisons with other market structures such as spot markets.
\appendices
\section{Preliminary Majorization Based Results}

Let $\VEC a = (a_1,a_2,\ldots,a_T)$ and $\VEC b = (b_1,b_2,\ldots,b_T)$ be two non-negative  vectors  arranged in non-increasing order (that is, $a_t \geq a_{t+1},$ and $b_t \geq b_{t+1}$). Then $\vec a \lessthan \vec b$ if
\bromalist
\item $\sum_{s=t}^T a_s \leq \sum_{s=t}^T b_s$, for $t=1,2,\ldots,T$, and 
\item $\sum_{s=1}^T a_s= \sum_{s=1}^T b_s$.
\end{list}
 The first condition above is equivalent to the following condition: Let $\VEC c$ be any rearrangement of $\VEC b$; then, for any  $S  \subset \{1,2,\ldots,T\}$, there exists $S' \subset \{1,2,\ldots,T\}$ of the same cardinality as $S$ such that $\sum_{s \in S'} a_s \leq \sum_{s \in S} c_s$.

\begin{definition}{\rm 
 We define a 1 unit {\em Robin Hood (RH) transfer} on $\VEC a$ as an operation that:
 
 \bromalist
 \item Selects indices $t,s$ such that $a_t > a_s$,
 \item Replaces $a_t$ by $a_t-1$ and $a_s$ by $a_s+1$.
 \item Rearranges the resulting vector in a non-increasing order.
 \end{list}}
\end{definition}

\begin{lemma}\label{lemma:1}
 Let $\VEC{\tilde a}$ be a vector obtained from $\VEC a$ after a 1 unit RH transfer. Then, $\VEC a \lessthan \VEC{\tilde a}$.
\end{lemma}
\proof Note that $\VEC{\tilde a}$ is a rearrangement of the vector $\VEC{\hat {a}} =(a_1,a_2,\ldots, a_t-1,\ldots, a_s +1,\ldots, a_T)$. In order to prove the lemma,  it suffices to show that for any set  $S  \subset \{1,2,\ldots,T\}$, there exists $S' \subset \{1,2,\ldots,T\}$ of the same cardinality as $S$ such that $\sum_{r \in S'} a_r \leq \sum_{r \in S} {\hat a}_r$.

For any subset $S$ of $\{1,2,\ldots,T\}$,
\bromalist
\item if $t \in S, s\in S$ (or if $t \notin S, s\notin S$), then $\sum_{r \in S} \hat a_r = \sum_{r \in S} a_r$. In this case, $S' =S$ satisfies our requirement.
\item if $t \notin S, s \in S$, then $\sum_{r \in S} \hat a_r = \sum_{r \in S} a_r +1$. In this case, $S' =S$ satisfies our requirement.
\item if $t \in S, s\notin S$, then $\sum_{r \in S} \hat a_r = a_t-1 + \sum_{r \in S\setminus\{t\}} \hat a_r \geq a_s +  \sum_{r \in S\setminus\{t\}} \hat a_r = a_s +  \sum_{r \in S\setminus\{t\}} a_r $. In this case, $S' =S\setminus \{t\} \cup \{s\}$ satisfies our requirement.
\end{list}
\endproof

\begin{lemma}\label{lemma:two}
Suppose $\VEC a \lessthan \VEC b$ \footnote{Recall that the vectors are arranged in non-increasing order.}. If for any $t$, $1 \leq t \leq T$, the following conditions hold:
\bromalist
\item $a_j=b_j$ for all $j < t$,
\item $a_t-a_T \leq 1$.
\end{list}
Then, (a) $a_t=b_t$, and (b) $\VEC a = \VEC b$.
\end{lemma}
\proof
$\VEC a \lessthan \VEC b$ and $a_j=b_j$ for all $j < t$ imply that 
\begin{align}
&\sum_{s=t+1}^T a_s \leq \sum_{s=t+1}^T b_s,  \notag \\
&\sum_{s=t}^T a_s = \sum_{s=t}^T b_s. \label{eq:lemma2eq}
\end{align}
Therefore, 
\begin{align}
a_t &= \sum_{s=t}^T a_s - \sum_{s=t+1}^T a_s \notag \\
&\geq \sum_{s=t}^T b_s-\sum_{s=t+1}^T b_s =b_t.
\end{align}
If $a_t > b_t$, then \[\sum_{s=t}^T b_s \leq (T-t+1)b_t \leq (T-t+1)(a_t-1).\] On the other hand, since $a_T \geq a_t-1$, \[\sum_{s=t}^T a_s >  a_t -1 + \sum_{s=t+1}^T a_s  \geq (T-t+1)(a_t-1).\]

 This implies that  $\sum_{s=t}^T a_s \neq \sum_{s=t}^T b_s$, which contradicts \eqref{eq:lemma2eq}. Thus, $b_t$ must be equal to $a_t$. Reapplying the first part of the lemma for $t+1$, then gives $a_{t+1}=b_{t+1}$. Proceeding sequentially till $T$ proves the second part of the lemma. 
\endproof

\begin{lemma} \label{lemma:3}
Let $\VEC a \lessthan \VEC b$ and $\VEC a \neq \VEC b$. Then, there exists a 1 unit RH operation on $\VEC a$ that gives a vector $\VEC{\tilde a } \neq \VEC a$ satisfying $\VEC a \lessthan \VEC{\tilde a} \lessthan \VEC b$.
\end{lemma}
\proof $\VEC a \lessthan \VEC b$ implies that for all $t$,
\begin{align}
&\sum_{s=1}^t b_s  = \sum_{s=1}^T b_s - \sum_{r=t+1}^T b_r \notag \\
&\leq \sum_{s=1}^T a_s - \sum_{r=t+1}^T a_r = \sum_{s=1}^t a_s \label{eq:28}
\end{align}
Let $t$ be the smallest index  such that $a_t \neq b_t$. Then, $b_t<a_t$ since $\sum_{i=1}^t b_i \leq \sum_{i=1}^t a_i$ and the two vectors have the same first $t-1$ elements. Let $s>t$ be the smallest index such that $a_t-a_s>1$.  Such $s$ must exist, otherwise Lemma \ref{lemma:two} would imply that $\VEC a = \VEC b$. Consider a 1 unit RH transfer from $t$ to $s$. Let $\VEC{\tilde a}$ be the resulting vector. Then, by Lemma \ref{lemma:1}, $\VEC a \lessthan \VEC{\tilde a}$. 

Also, if $k$ is the number of elements of $\VEC a$ with value equal to  $a_t$, then the number of elements of $\VEC{\tilde a}$ with value equal to $a_t$ is $k-1$. Therefore, $\VEC{\tilde a } \neq \VEC a$.
 
Further, it is clear that $\VEC a$ and $\VEC {\tilde{a}}$ have the same first $t-1$ elements (since the RH operation depleted 1 unit from $a_t$ and added it to $a_s<a_t-1$, the non-increasing rearrangement would not change the $t-1$ highest elements.) Similarly, $\VEC a$ and $\VEC{\tilde{a}}$ have the same elements from index $s$ to $T$. Further, from the definition of $s$, $a_j \geq a_t -1$ for $t \leq j <s$. Since,  $\tilde{a}_t,\tilde{a}_{t+1} \ldots, \tilde{a}_{s-1}$, must be a rearrangement of $a_{t}-1,a_{t+1}\ldots, a_{s-1}$, it follows that $\tilde a_j \geq a_t -1$ for $t \leq j <s$.

We now prove that $\VEC{\tilde a} \lessthan \VEC b$. Clearly, 
\[ \sum_{s=1}^T \tilde a_s = \sum_{s=1}^T a_s = \sum_{s=1}^T b_s. \]
Therefore, it suffices to prove that for all $r$, $\sum_{i=1}^r b_i \leq \sum_{i=1}^r \tilde a_i$.
\bromalist
\item If $r < t$ or $r \geq s$, then $\sum_{i=1}^r \tilde a_i = \sum_{i=1}^r a_i \geq \sum_{i=1}^r b_i$, by \eqref{eq:28}.
\item If $t \leq r < s$, then 
\begin{align}
&\sum_{i=1}^r \tilde a_i = \sum_{i=1}^{t-1}  a_i + \sum_{i=t}^{r}  \tilde a_i  \notag\\
& =\sum_{i=1}^{t-1}  b_i +\sum_{i=t}^{r}  \tilde a_i \geq \sum_{i=1}^{t-1}  b_i +(r-t+1)(a_t-1), \label{eq:one}
\end{align}
where the first and second equalities are true because the first $t-1$ elements of $\VEC a, \VEC{\tilde a}$ and $\VEC b$ are the same, and  last inequality follows from the fact that for $t \leq j <s$, $\tilde{a}_i \geq (a_{t}-1)$.
Moreover, 
\begin{align}
&\sum_{i=1}^rb_i = \sum_{i=1}^{t-1}  b_i+  \sum_{i=t}^r b_i \notag \\
&\leq \sum_{i=1}^{t-1}  b_i+ (r-t+1)b_t \leq \sum_{i=1}^{t-1}  b_i+ (r-t+1)(a_t-1),\label{eq:two}
\end{align}
where the last inequality follows from $b_t < a_t$. Equations \eqref{eq:one} and \eqref{eq:two} imply that
$\sum_{i=1}^r b_i \leq \sum_{i=1}^r \tilde a_i$,  and hence $\VEC{\tilde a} \lessthan \VEC b$. 
\end{list}
\endproof

\begin{claim}
 Let  $\VEC a \lessthan \VEC b$ \footnote{Recall that the vectors are arranged in non-increasing order.} with $\VEC a \neq \VEC b$. Then, there exists a finite sequence of 1 unit RH transfers that can be applied on $\VEC a$ to get $\VEC b$.
\end{claim}
\proof The claim is established using Lemmas \ref{lemma:two} and \ref{lemma:3}.
Let $\VEC a^0 = \VEC a$.
\bromalist
\item For $n =1,2,\ldots$, if $\VEC a^{n-1} \neq \VEC b$,  use Lemma \ref{lemma:3} to construct $\VEC a^n \neq \VEC a^{n-1}$ such that $\VEC a^{n-1} \lessthan \VEC a^n \lessthan \VEC b$. Then, $\VEC a^n \neq \VEC a^m$ for any $m<n-1$ 
(otherwise, we would have $\VEC a^m=\VEC a^n \lessthan \VEC a^{n-1} \lessthan \VEC a^n \implies \VEC a^n=\VEC a^{n-1}$).
\item If  $\VEC a^{n-1} = \VEC b$, stop.
\end{list}

Since there are only finitely many non-negative integer valued vectors that majorize $\VEC b$, this procedure must  stop after finite number of steps and it can do so only if $\VEC a^{n-1}=\VEC b$. This proves the claim.                                 
\endproof

\section{Proof of Lemma \ref{lemma:unitrate}} \label{sec:unitrate}
If $u_t = \sum_{j=1}^{m} \nu^j_t$, $t =1,\ldots, T$, where 
\begin{align}
& \sum_t \nu^j_t = k^i+1, ~\nu^j_t \in \{0,1\} ~~\mbox{for $j=1,2,\ldots,\ell$}, \notag \\
 & \sum_t \nu^j_t = k^i, ~\nu^j_t \in \{0,1\} ~~\mbox{for $j=\ell+1,\ldots, m$},
\end{align}
then clearly $u_t \in \{0,1,\ldots, m\}$ and $\sum_{t=1}^T u_t = km+\ell =E$.

To prove the converse,  first assume $\ell>0$. Consider an allocation $u_t \in \{0,1,\ldots, m\}$ satisfying $\sum_{t=1}^T u_t = km+\ell =E$. Then, $u_t > 0$ for at least $k+1$ time slots and $u_t= m$ for at most $k$ time slots. Pick $k+1$ time slots with largest value of $u_t$. Define $\nu^1_t =1$ at the selected slots and $0$ otherwise. Let $B_t = u_t - \nu^1_t$. Then, $B_t \in \{0,1,\ldots,m-1\}$ and $\sum_t B_t = k(m-1) +(\ell-1)$. If $\ell=0$, then $u_t > 0$ for at least $k$ time slots and $u_t = m$ for at most $k$ time slots. Pick $k$ time slots with largest value of $u_t$. Define $\nu^1_t =1$ at the selected slots and $0$ otherwise. Let $B_t = u_t-\nu^1_t$. Then, $B_t \in \{0,1,\ldots,m-1\}$ and $\sum_t B_t = k(m-1) $. Thus, we can always write $u_t$ as
\begin{equation}
u_t = \nu^1_t + B_t,
\end{equation}
where $\nu^1_t \in \{0,1\}$, $\sum_t \nu^1_t = k + \mathds{1}_{\{\ell >0\}}$, $B_t\in \{0,1,\ldots,m-1\}$ and \[\sum_t B_t = k(m-1) + (\ell-1)^+.\]
Now, if $(\ell-1)>0$, then $B_t >0$ for at least $k+1$ slots and $B_t = m-1$ for at most $k$ slots. Pick $k+1$ time slots with largest value of $B_t$. Define $\nu^2_t =1$ at the selected slots and $0$ otherwise. Let $C_t = B_t - \nu^2_t$. Then, $C_t \in \{0,1,\ldots,m-2\}$ and $\sum_t C_t = k(m-2) +(\ell-2)$. If $(\ell-1)^+ =0$, then $B_t > 0$ for at least $k$ time slots and $B_t = m$ for at most $k$ time slots. Pick $k$ time slots with largest value of $B_t$. Define $\nu^2_t =1$ at the selected slots and $0$ otherwise. Let $C_t = B_t - \nu^2_t$. Then, $C_t \in \{0,1,\ldots,m-2\}$ and $\sum_t C_t = k(m-2) $. Thus, we can write $u_t$ as
\begin{equation}
u_t = \nu^1_t + \nu^2_t +  C_t,
\end{equation}
where $\nu^2_t \in \{0,1\}$, $\sum_t \nu^2_t = k + \mathds{1}_{\{\ell >1\}}$,$C_t \in \{0,1,\ldots,m-2\}$ and \[\sum_t C_t = k(m-2) + (\ell-2)^+.\] Continuing sequentially, we can decompose $u_t$ into $\nu^1_t,\ldots, \nu^m_t$. 

\section{Proof of Theorem \ref{thm:two}} \label{sec:two}
Recall that \[ \mathcal{D}(\mathcal{C}) = \{ (E^j,1)| j=1,2,\ldots, N\},\]
\[ d_t := \sum_{j=1}^N \mathds{1}_{\{ E^j \geq t\}}.\]
We require with the following intermediate result.
\begin{lemma}\label{lemma:one}
For a given collection of services $\mathcal{D}(\mathcal{C})$, if $\VEC a$ is an exactly adequate supply, then any  supply $\VEC b$ satisfying $\VEC a \lessthan \VEC b$  is also exactly adequate.
\end{lemma}
\proof Without loss of generality, we will assume that $\VEC a$ and $\VEC b$ are arranged in non-increasing order.
Since $\VEC b$ can be obtained from $\VEC a$ by a sequence of 1 unit RH transfers (see Claim 1, Appendix A), we simply need to prove that a 1 unit RH transfer preserves exact adequacy. Consider an exactly adequate supply vector $\VEC a$ and let $\nu^j_t$, $j=1,\ldots,N, t=1,\ldots,T$ be the corresponding allocations. Consider  a 1 unit RH transfer  from time $t$ to $s$ (without rearrangement) that gives a new supply vector $\VEC{\tilde a}$.   Let $i$ be a service for which $\nu^i_t =1$ but $\nu^i_s =0$. Such an $i$ must exist because $\sum_j \nu^j_t = a_t > a_s = \sum_j \nu^j_s$. Under the new supply, the allocations 
\begin{align}
\tilde{\nu^j_r} =\left \{\begin{array}{ll} & \nu^j_r ~~r\neq t,s \mbox{~or~} j\neq i\\
                                      &0 ~~r=t,j=i \\
                                      &1 ~~r=s,j=i \\
                                        \end{array}
                                       \right.
\end{align} 
establish exact adequacy.                      
\endproof

{\bf Proof of Theorem \ref{thm:two} (a):}  \  Observe that  a supply vector $\vec p =\VEC d $ is exactly adequate: the allocation function $\nu^j_t = \mathds{1}_{\{ E^j \geq t\}}$ meets the exact adequacy requirements. Therefore, by Lemma \ref{lemma:one}, any $\VEC p$ satisfying $\VEC p \morethan \VEC d$ is also exactly adequate. 
 
To prove necessity, suppose $\VEC p$ is exactly adequate and $\nu^j_t$, $j=1,\ldots,N, t=1,\ldots,T$ are the corresponding allocations. Consider any set $\mathcal{S} \subset \{1,2,\ldots,T\}$ of cardinality $s$. Then, because $\nu^j_t \in \{0,1\}$ and $\sum_{t=1}^T \nu^j_t = E^j$, it follows that
\begin{equation} \label{eq:necc1}  \sum_{t \in \mathcal{S}} \nu^j_t \leq \min(s,E^j) = \sum_{t=1}^s \mathds{1}_{\{t \leq E^j\}}. \end{equation}
Summing over $j$,
\begin{align}
 &~~\sum_{j=1}^N\sum_{t \in \mathcal{S}} \nu^j_t \leq \sum_{j=1}^N\sum_{t =1}^s \mathds{1}_{\{t \leq E^j\}} \notag \\
 &\implies \sum_{t \in \mathcal{S}}\sum_{j=1}^N \nu^j_t \leq \sum_{t=1}^s \sum_{j=1}^N\mathds{1}_{\{t \leq E^j\}} \notag \\
 &\implies \sum_{t \in \mathcal{S}} p_t \leq \sum_{t=1}^s d_t  \quad \mbox{for all $\mathcal{S}$ with $|\mathcal{S}| =s$}\notag \\
 &\implies   \sum_{t=1}^s p^{\downarrow}_t\leq \sum_{t=1}^s d_t\label{eq:necc2}
\end{align}
For $s=T$, the inequality in \eqref{eq:necc1} becomes an equality resulting in an equality in \eqref{eq:necc2}. 
\endproof

{\bf Proof of Theorem \ref{thm:two} (b):}
To prove necessity, suppose $\VEC p$ is adequate and $\nu^j_t$, $j=1,\ldots,N, t=1,\ldots,T$ are the corresponding allocations. Then, for any set $\mathcal{S} \subset \{1,\ldots,T\}$ of cardinality $s$
\begin{equation}
\sum_{t \in \mathcal{S}} p_t \geq \sum_{t \in \mathcal{S}} \sum_{j=1}^N \nu^j_t\label{eq:app4_eq1}
\end{equation}
 Further, because $\nu^j_t \in \{0,1\}$ and $\sum_{t=1}^T \nu^j_t = E^j$, it follows that
\begin{align}  
 &\sum_{t \in \mathcal{S}} \nu^j_t \geq \max(E^j-s+1,0) = \sum_{t=s}^T \mathds{1}_{\{t \leq E^j\}} \notag  
 \end{align}
Summing over $j$,
\begin{align}
 &~~\sum_{j=1}^N\sum_{t \in \mathcal{S}} \nu^j_t \geq \sum_{j=1}^N\sum_{t=s}^T \mathds{1}_{\{t \leq E^j\}} \notag \\
 &\implies \sum_{t \in \mathcal{S}}\sum_{j=1}^N \nu^j_t \geq \sum_{t=s}^T \sum_{j=1}^N\mathds{1}_{\{t \leq E^j\}} \notag \\
 &\implies \sum_{t \in \mathcal{S}}\sum_{j=1}^N \nu^j_t \geq  \sum_{t=s}^T d_t. \label{eq:app4_eq2}
\end{align}
Combining \eqref{eq:app4_eq1} and \eqref{eq:app4_eq2} proves the necessity conditions.

To prove sufficiency, let $\Delta = \sum_{t=1}^T p_t - \sum_{t=1}^T d_t$. Then, $\Delta \geq 0$. Consider a new demand duration vector $\VEC d^{\Delta}$ defined as $d^{\Delta}_t := d_t + \Delta\mathds{1}_{\{t \leq 1\}}$. This new demand duration vector corresponds to the original set of unit rate services augmented  with $\Delta$ ``fictitious unit rate services'' with $E^j=1$. It is easy to see that $\VEC p \morethan^w \VEC d$ implies $\VEC p \morethan \VEC d^{\Delta}$. Therefore, by part (a) of Theorem \ref{thm:two}, $\VEC p$ is exactly adequate for the augmented set of services which implies that it must be adequate for the original set of services $\mathcal{D}(\mathcal{C})$.
\endproof
\section{Proof of Theorem \ref{thm:allocation}}\label{sec:app_3}
Since $\VEC p$ is adequate, there must exist at least one set of allocations $b^j_t, j=1\ldots,N,t=1,\ldots,T$
\[ \sum_{t=1}^T b^j_t = E^j, \quad  \sum_{j=1}^N b^j_t \leq p_t.\]
Let $\mathcal{B}_1$ be the subset  of services in $\mathcal{D}(\mathcal{C})$ that are served at time $1$, that is, for which $b^j_1 =1$. Similarly, let $\mathcal{V}_1$ be the subset of services in  $\mathcal{D}(\mathcal{C})$ for which $\nu^j_1$ defined in \eqref{eq:allocation1} is $1$. It is easy to check that $|\mathcal{B}_1| \leq |\mathcal{V}_1|$.

 If $\mathcal{V}_1 \neq \mathcal{B}_1$,  but the two sets have the same cardinality, pick a service $i \in \mathcal{V}_1\setminus\mathcal{B}_1$ and $j \in \mathcal{B}_1\setminus\mathcal{V}_1$. Then, $\pi(i)<\pi(j)$ which implies that $E^i\geq E^j$. Therefore, there must exist a time $s>1$ such that $b^i_s =1$ but $b^j_s=0$. Consider a new set of allocations  obtained by swapping the load $i$  at time $s$ and the load $j$ at time $1$, that is, the new allocation rule $b^k_t(1), k=1,\ldots,N, t=1,\ldots,T$ is identical to $b^k_t$ except that for $t=1,s$
\[ b^i_t(1) = b^j_t \quad \mbox{and} \quad b^j_t(1)=b^i_t. \]
It is straightforward to establish that the new allocation rule satisfies the adequacy requirements. One can proceed  with this swapping argument one service at a time until the set of services being served at time $1$ is equal to $\mathcal{V}_1$.

In the case where $|\mathcal{B}_1| < |\mathcal{V}_1|$, pick $i \in \mathcal{V}_1\setminus\mathcal{B}_1$. There must be a time $s>1$ such that $b^i_s =1$. Consider a new allocation that is identical to $b^k_t$ except that $b^i_1=1$ and $b^i_s=0$. The new set of services served at time $1$ is now $\mathcal{B} \cup \{i\}$. Keep repeating this argument, until $|\mathcal{V}_1|$ number of services are being served at time $1$. Then, use the swapping argument from above to transform the set of services being served at time $1$ to $\mathcal{V}_1$.

For time $t>1$, the same swapping argument works by replacing $E^j$  with the energy needed by service $j$ from time $t$ to the final time, which is $E^j - \sum_{s=1}^{t-1} \nu^j_s$.
\section{Proof of Lemma \ref{lemma:LB}} \label{sec:LB}
 
Let $\mathcal{S}$ be any non-empty subset of  $\{1,2,\ldots,T\}$ and let $|\mathcal{S}|$ denote the cardinality of $\mathcal{S}$.  
\begin{align}
&\sum_{t=1}^T (q_t-p_t) \geq \sum_{t \in \mathcal{S}} (q_t-p_t) \notag \\
& \geq \sum_{s=T-|\mathcal{S}|+1}^T d_s - \sum_{t \in \mathcal{S}} p_t, \label{eq:lb2}
\end{align}
where we used the fact that $\vec{d} \lessthan^w \vec q$ in \eqref{eq:lb2}. Since \eqref{eq:lb2} holds for any $\mathcal{S}$, it follows that  
\begin{align}
&\sum_{t=1}^T (q_t-p_t) \geq \max_{\mathcal{S} \subset \{1,\ldots,T\}} \left( \sum_{s=T-|\mathcal{S}|+1}^T d_s - \sum_{t \in \mathcal{S}} p_t\right) \notag \\
& = \max_{1 \leq r \leq T} \left[ \max_{\mathcal{S} \subset \{1,\ldots,T\}, |\mathcal{S}|=T-r+1} \left( \sum_{s=T-|\mathcal{S}|+1}^T d_s - \sum_{t \in \mathcal{S}} p_t\right)\right]  \notag \\
& = \max_{1 \leq r \leq T} \left[ \sum_{s=r}^T d_s -  \min_{\mathcal{S} \subset \{1,\ldots,T\}, |\mathcal{S}|=T-r+1} \sum_{t \in \mathcal{S}} p_t \right] \notag \\
& = \max_{1 \leq r \leq T} \left[ \sum_{s=r}^T d_s -  \sum_{s=r}^T p^{\downarrow}_s \right]. \label{eq:lb1}
\end{align}
 \eqref{eq:lb1} and the fact that $\sum_{t=1}^T(q_t-p_t) \geq 0$ proves the lemma.
\section{Proof of Theorem \ref{thm:main_thm}}\label{sec:main_thm}
Recall that $q_t = y_t +r_t + g^*_t(\mathcal{C}, \vec y, r_1,\ldots,r_t)$, where $g^*(\cdot)$ is defined as the solution of the minimization problem in \eqref{eq:rt_purch}.

{\bf Proof of Part 1.} From the definition of $a_1$, we have  $q_1 \geq d_T$ which implies the following majorization relation between the one-dimensional vectors $q_1$ and $d_T$: $q_1 \morethan^w d_T$. 

We now proceed by a induction argument. Assume that the $t-1$ dimensional majorization relation $(q_1,\ldots,q_{t-1}) \morethan^w (d_{T-t+2},\ldots,d_T)$ is true. Then,  the  optimization problem of \eqref{eq:rt_purch} at time $t$ is equivalent to the following linear program
\begin{align}
\min_{a_t \geq 0} a_t  \notag
\end{align}
 subject to
\begin{eqnarray*}
& & y_t+r_t+a_t \geq d_T \\
& &q_s+y_t+r_t+ a_t \geq d_T + d_{T-1}, \quad \mbox{for all~}s <t \\
& & \sum_{s=s_1,s_2,\ldots,s_k}q_s + (y_t+r_t+a_t) \geq d_{T-k}+\cdots+ d_T, \\ 
&&\mbox{for all~} s_1 < \cdots < s_k <t  \mbox{ and~} k \leq t-1
\end{eqnarray*}
It is straightforward to see that the above linear program has a non-negative integer solution. (Recall that all variables in the linear program are non-negative integers.)
The resulting $t-$ dimensional vector $(q_1,\ldots,q_t)$ (where $q_t =y_t+r_t+a_t$) satisfies the $t-$ dimensional majorization relation $(q_1,\ldots,q_{t}) \morethan^w (d_{T-t+1},\ldots,d_T)$. The induction argument is now complete and we conclude the $T-$ dimensional majorization relation $\vec q \morethan^w \vec d$ is true.

{\bf Proof of Part 2.} This is an immediate consequence of Theorem \ref{thm:allocation} and Lemma \ref{lemma:unitrate2}.

{\bf Proof of Part 3.} We first define $T+1$ vectors $\vec p(0), \vec p(1), \ldots, \vec p(T)$ as follows:
\begin{align}
&\vec p(0) := \vec y + \vec r, \notag \\
& \vec p(1) := \vec y + \vec r + (a_1,0,\ldots,0) \notag \\
&\vec p(t) := \vec y + \vec r + (a_1,a_2,\ldots, a_t, 0 \ldots 0)  ~~ \mbox{ for $1< t <T$},\notag \\
&\vec p(T) = \vec y + \vec r+ (a_1,\ldots, a_T) = \vec q,
\end{align}
where $a_1,a_2,\ldots,a_T$ are the real-time purchases under the policy $g^*$ obtained by the solving the minimization in \eqref{eq:rt_purch}.

We also define the following functions:

\begin{enumerate}
\item  For any vector $\VEC x = (x_1,\ldots,x_T)$, 
\[F(\VEC x) := \left[\max_{\mathcal{S} \subset \{1,2,\ldots,T\}} \left(\sum_{ s= T- |\mathcal{S}|+1}^T d_s - \sum_{r \in \mathcal{S}} x_r\right)\right]^+.\]

\item $\mathcal{S}(\vec x)$ be any subset of $\{1,2,\ldots,T\}$ that achieves the maximum in the definition of $F(\vec x)$.
\end{enumerate}

We will use the following result.
\begin{claim} \label{claim:2}
1. $F(\vec p(T)) =0$ and \[F(\vec p(0)) = L(\mathcal{C}, \vec y, \vec r) = \left[\max_{1 \leq r \leq T} \sum_{ s = r}^{T} (d_s -p^{\downarrow}_s)\right]^+,\] where $\vec p =\vec y+\vec r$. \\
2. For $t=2,\ldots,T$, 
\begin{equation} 
 F(\vec p(t)) = F(\vec p(t-1)) -a_t. \label{eq:F_eq}
\end{equation}
\end{claim}
 
 Suppose Claim \ref{claim:2} is true. Then, combining \eqref{eq:F_eq} for all $t$ gives
\begin{align}
&F(\VEC p(T)) = F(\VEC p) - \sum_{t=1}^T a_t \notag \\
&\implies \quad  \sum_{t=1}^T a_t = F(\VEC p) - F(\VEC p(T)) =L(\mathcal{C}, \vec y, \vec r) , \label{eq:costbound3}
\end{align}
and this proves part 3 of the theorem. 

We now prove Claim \ref{claim:2}. By part 1 of Theorem \ref{thm:main_thm}, $\vec p(T) = \vec q \morethan^w \vec d$. Therefore, for any set $\mathcal{S}$,
\[ \sum_{ s= T- |\mathcal{S}|+1}^T d_s \leq \sum_{r \in \mathcal{S}} q_r. \]
This implies that $F(\vec p(T)) =0$.

For $\vec p(0)$,
\begin{align}
&F(\VEC p(0)) =\left[\max_{\mathcal{S} \subset \{1,2,\ldots,T\}} \left(\sum_{ s= T- |\mathcal{S}|+1}^T d_s - \sum_{t \in \mathcal{S}} p_t(0)\right)\right]^+ \notag \\
&=  \left[\max_{1 \leq r \leq T}  \max_{\substack{\mathcal{S} \subset \{1,\ldots,T\},\\ |\mathcal{S}|=T-r+1}} \left( \sum_{s=T-|\mathcal{S}|+1}^T d_s - \sum_{t \in \mathcal{S}} p_t(0)\right)\right]^+  \notag \\
& = \left[\max_{1 \leq r \leq T}  \sum_{s=r}^T d_s -  \min_{\substack{\mathcal{S} \subset \{1,\ldots,T\}, \\ |\mathcal{S}|=T-r+1}} \sum_{t \in \mathcal{S}} p_t(0) \right]^+ \notag \\
& =  \left[\max_{1 \leq r \leq T} \sum_{s=r}^T d_s -  \sum_{s=r}^T p^{\downarrow}_s(0) \right]^+ \notag \\
&= L(\mathcal{C}, \vec y, \vec r). \notag
\end{align}
This proves the first part of the claim.

For the second part, if $a_t =0$, then $\vec p(t) = \vec p(t-1)$ and therefore \eqref{eq:F_eq} trivially holds. We consider the case where $a_t > 0$. Recall from the proof of Part 1 of Theorem \ref{thm:main_thm} that $a_t$ is the solution of the linear program
\begin{align}
\min_{a_t \geq 0} a_t  \notag
\end{align}
 subject to
\begin{eqnarray*}
& & y_t+r_t+a_t \geq d_T \\
& &q_s+y_t+r_t+a_t \geq d_T + d_{T-1}, \quad \mbox{for all~}s <t \\
& & \sum_{s=s_1,s_2,\ldots,s_k}q_s + (y_t+r_t+a_t) \geq d_{T-k}+\cdots+ d_T, \\ 
&&\mbox{for all~} s_1 < \cdots < s_k <t  \mbox{ and~} k \leq t-1
\end{eqnarray*}
(Also recall that for $s <t$, $q_s = y_s+r_s + a_s = p_s(t-1)$.)
The constraints in the above program can be written as:
\[\sum_{s=T-|\mathcal{T}|+1}^T d_s - \sum_{s \in \mathcal{T}} p_s(t-1)-a_t  \leq 0 ,\]
for all $\mathcal{T} \subset \{1,\ldots,t\}$ that contain $t$.

Let $\mathds{T}$ be the collection of all non-empty subsets of $\{1,2,\ldots,t\}$. For any $\mathcal{T} \in \mathds{T}$ that does not contain $t$ (that is, $\mathcal{T} \subset \{1,\ldots,t-1\}$), the real-time purchase decisions $a_1,a_2,\ldots,a_{t-1}$ chosen according to \eqref{eq:rt_purch} ensure that 
 \[\sum_{s=T-|\mathcal{T}|+1}^T d_s - \sum_{s \in \mathcal{T}} p_s(t-1)  \leq 0.\]
However, $a_t > 0$ implies that there must be some set $\mathcal{T} \subset \{1,\ldots,t\}$ with $t \in \mathcal{T}$ for which
\begin{equation} \label{eq:F_pf}
  \sum_{s=T-|\mathcal{T}|+1}^T d_s - \sum_{s \in \mathcal{T}} p_s(t-1)  >0.
\end{equation}
If no such $\mathcal{T}$ existed, then $a_t$ should have been $0$.
Consider the following maximization 
\begin{equation}
\max_{\mathcal{T} \in \mathds{T}} \sum_{s=T-|\mathcal{T}|+1}^T d_s - \sum_{s \in \mathcal{T}} p_s(t-1). \label{eq:maxT}
\end{equation}
and let $\mathcal{T}^* $ be a maximizing set.
 From \eqref{eq:F_pf}, it follows that the maximum value in \eqref{eq:maxT} is positive and that $t \in \mathcal{T}^*$. 
 
 We now proceed as follows:
 
 {\bf Fact 1: $F(\vec p(t-1))$ is maximized by some $\mathcal{S}(\vec p(t-1)) \supset \mathcal{T}^*$. }
 
To prove this fact, consider any $\mathcal{S}$ and write it as  $ \left(\mathcal{S} \cap \mathcal{T}^*\right)  \cup \mathcal{S}\setminus \mathcal{T}^*$. Then,
\begin{align}
&\sum_{ s= T- |\mathcal{S}|+1}^T d_s - \sum_{s \in \mathcal{S}} p_s(t-1) \notag \\
&=\left[\sum_{ s= T- |\mathcal{S}\cap \mathcal{T}^*|+1}^T d_s - \sum_{s \in \mathcal{S}\cap \mathcal{T}^*} p_s(t-1) \right] \notag \\
&+\left[\sum_{ s= T- |\mathcal{S}|+1}^{ T- |\mathcal{S}\cap \mathcal{T}^*|} d_s - \sum_{s \in \mathcal{S}\setminus \mathcal{T}^*} p_s(t-1) \right]. \label{eq:fact1}
\end{align}
Now consider a $\mathcal{S}' := \mathcal{T}^* \cup \mathcal{S}\setminus \mathcal{T}^* $. For this set,
\begin{align}
&\sum_{ s= T- |\mathcal{S}'|+1}^T d_s - \sum_{s \in \mathcal{S}'} p_s(t-1) \notag \\
&=\left[\sum_{ s= T- | \mathcal{T}^*|+1}^T d_s - \sum_{s \in  \mathcal{T}^*} p_s(t-1) \right] \notag \\
&+\left[\sum_{ s= T- |\mathcal{S}'|+1}^{ T- | \mathcal{T}^*|} d_s - \sum_{s \in \mathcal{S}\setminus \mathcal{T}^*} p_s(t-1) \right]. \label{eq:fact2}
\end{align}
Compare each of the two bracketed terms in \eqref{eq:fact1} with the corresponding term in \eqref{eq:fact2}. The first term in \eqref{eq:fact1} is no greater  than the corresponding term in \eqref{eq:fact2} because $\mathcal{T}^*$ achieves the maximum in \eqref{eq:maxT}.  The second term in \eqref{eq:fact1} is no greater than the corresponding term in \eqref{eq:fact2} because $\vec d $ is non-increasing in its index ($d_t \geq d_{t+1}$). Thus, $F(\vec p(t-1))$ is maximized by some $\mathcal{S}(\vec p(t-1)) \supset \mathcal{T}^*$.

{\bf Fact 2: $F(\vec p(t))$ is maximized by some $\mathcal{S}(\vec p(t)) \supset \mathcal{T}^*$. }
Note that the $\vec p(t-1)$ and $\vec p(t)$ are the same except for the $t^{th}$ element.
The real-time purchase at time $t$,  $a_t$ the smallest value added to $\vec p(t-1)$ that ensures that 
\[ (d_{T-t+1},\ldots, d_T ) \lessthan^w (p_1(t),p_2(t),\ldots, p_t(t)).\] Therefore, for any $\mathcal{T} \subset \{1,\ldots,t\}$, 
\[\sum_{s=T-|\mathcal{T}|+1}^T d_s - \sum_{s \in \mathcal{T}} p_s(t)  \leq 0,\]
with equality for $\mathcal{T} = \mathcal{T}^*$.
We can now repeat the arguments used in proving Fact 1 to establish Fact 2.

{\bf Fact 3: Both $F(\vec p(t-1))$ and $F(\vec p(t))$) are maximized by a set $\mathcal{U} = \mathcal{T}^* \cup \mathcal{V}^*$, where $\mathcal{T}^* \cap \mathcal{V}^* = \emptyset$.}

Because of Fact 1,  $F(\vec p(t-1))$ can be written as
\begin{align}
&\left[\sum_{ s= T- | \mathcal{T}^*|+1}^T d_s - \sum_{s \in  \mathcal{T}^*} p_s(t-1) \right]  + \notag \\ &\max_{\mathcal{V} \subset \{1,\ldots,T\} \setminus \mathcal{T}^*} 
\left[ \sum_{ s= T- |\mathcal{V}| - | \mathcal{T}^*|+1}^{ T- | \mathcal{T}^*|} d_s - \sum_{s \in \mathcal{V}} p_s(t-1) \right] \label{eq:fact3a}
\end{align}
Similarly, for $F(\vec p(t))$,
\begin{align}
&\left[\sum_{ s= T- | \mathcal{T}^*|+1}^T d_s - \sum_{s \in  \mathcal{T}^*} p_s(t) \right]  + \notag \\ &\max_{\mathcal{V} \subset \{1,\ldots,T\} \setminus \mathcal{T}^*} 
\left[ \sum_{ s= T- |\mathcal{V}| - | \mathcal{T}^*|+1}^{ T- | \mathcal{T}^*|} d_s - \sum_{s \in \mathcal{V}} p_s(t) \right]\label{eq:fact3b}
\end{align}
Recall that $\vec p(t-1)$ and $\vec p(t)$ differ only in the $t^{th}$ element and that $t \in \mathcal{T}^*$. Therefore, the maximizations in \eqref{eq:fact3a} and \eqref{eq:fact3b} are the same. This establishes Fact 3.

Because of Fact 3,
\begin{align}
F(\vec p(t)) &= \left[\sum_{ s= T- | \mathcal{T}^*|+1}^T d_s - \sum_{s \in  \mathcal{T}^*} p_s(t) \right]   \notag \\ 
& +\left[ \sum_{ s= T- |\mathcal{V}^*| - | \mathcal{T}^*|+1}^{ T- | \mathcal{T}^*|} d_s - \sum_{s \in \mathcal{V}^*} p_s(t) \right] \notag \\
&= \left[\sum_{ s= T- | \mathcal{T}^*|+1}^T d_s - \sum_{s \in  \mathcal{T}^*} p_s(t-1) \right] - a_t   \notag \\ 
& +\left[ \sum_{ s= T- |\mathcal{V}^*| - | \mathcal{T}^*|+1}^{ T- | \mathcal{T}^*|} d_s - \sum_{s \in \mathcal{V}^*} p_s(t-1) \right] \notag \\
&= F(\vec p(t-1)) -a_t.
\end{align}
This proves the second part of Claim \ref{claim:2}.

\section{Proof of Theorem \ref{thm:opt}} \label{sec:opt}
{\bf Part 1.} In order to prove the convexity of the objective function in \eqref{eq:max_prob2} , it suffices to prove the convexity of 
\[ \max_{1 \leq t \leq T} \left(\sum_{ s= t}^T (d_s -(\VEC r +\VEC y)^{\downarrow}_s)\right)^+\]
for each possible realization $\VEC r$ of the renewable supply. Since maximum of convex functions is convex, it is sufficient to prove that 
\begin{equation} \label{eq:convexity1} \left(\sum_{ s = t}^T (d_s -(\VEC r +\VEC y)^{\downarrow}_s)\right)^+
\end{equation}
is convex for all $t$. \eqref{eq:convexity1} can be written as
\begin{align}
&= \sum_{ s =t}^{T} d_s - \min_{\substack{\mathcal{S} \subset \{1,2,\ldots,T\},\\ |\mathcal{S}| = T-t+1}} \left( \sum_{k \in \mathcal{S}} (r_k +y_k)\right) \notag \\
&=\sum_{ s =t}^{T} d_s + \max_{\substack{\mathcal{S} \subset \{1,2,\ldots,T\},\\ |\mathcal{S}| = T-t+1 }} \left( -\sum_{k \in \mathcal{S}} (r_k +y_k)\right) \notag \\
&= \max_{\substack{\mathcal{S} \subset \{1,2,\ldots,T\},\\ |\mathcal{S}| = T-t+1 }}\left(\sum_{ s =t}^{T} d_s  -\sum_{k \in \mathcal{S}} (r_k +y_k)\right) \label{eq:convexity2}
\end{align}
Since \eqref{eq:convexity2} is a maximum of affine functions of $ \vec d, \VEC y$, it implies that it is convex in $\vec d, \VEC y$. This completes the proof pf Part 1.

{\bf Part 2.} The construction of integer solutions in \eqref{eq:approx} implies that 
\begin{align}
&\sum_{t=1}^T \pi(t,1)(d^a_t-d^a_{t+1}) \geq\sum_{t=1}^T \pi(t,1)(d^c_t-d^c_{t+1}-1)\notag \\
& = \sum_{t=1}^T \pi(t,1)(d^c_t-d^c_{t+1})  - \sum_{t=1}^T\pi(t,1), \label{eq:bound1}
\end{align}
and 
\begin{equation} \label{eq:bound2} c^{da}\sum_{t=1}^T y^a_t \leq  c^{da}\sum_{t=1}^T (y^c_t+1) = c^{da}\sum_{t=1}^T y^c_t + c^{da}T 
\end{equation}
Thus, the decrease in revenue in changing $(\vec d^c, \vec y^c)$ to $(\vec d^a, \vec y^a)$ is less than $\sum_{t=1}^T \pi(t,1)$ while the increase in day-ahead energy cost is less than $c^{da}T$. 

Because, $\vec d^a \leq \vec d^c$ and $\vec y^a \geq \vec y^c$, it follows that for each realization $\vec r$ of $\VEC R$, 
\begin{align}
&\left[\max_{1 \leq t \leq T} \sum_{s = t}^{T}(d^a_s -(\vec y^a +\vec r)^{\downarrow}_s)\right]^+  \notag \\
& \leq \left[\max_{1 \leq t \leq T} \sum_{ s = t}^{T} (d^c_s -(\vec y^c +\vec r)^{\downarrow}_s)\right]^+. \label{eq:bound3}
\end{align}
Combining \eqref{eq:bound1},\eqref{eq:bound2} and \eqref{eq:bound3}, it follows that 
\[ J(\vec d^c, \vec y^c) - J(\vec d^a, \vec y^a) \leq c^{da}T + \sum_{t=1}^T \pi(t,1). \]
Finally, observing  that 
\[J(\vec d^*, \vec y^*) - J(\vec d^a, \vec y^a)  \leq J(\vec d^c, \vec y^c) - J(\vec d^a, \vec y^a),\]
proves \eqref{eq:gap}.

\bibliographystyle{IEEEtran}
\bibliography{timedif}

\end{document}